\documentclass[lettersize,journal]{IEEEtran}
\usepackage{amsmath,amsfonts}
\usepackage{algorithmic}
\usepackage{algorithm}
\usepackage{array}
\usepackage[caption=false,font=normalsize,labelfont=sf,textfont=sf]{subfig}
\usepackage{textcomp}
\usepackage{stfloats}
\usepackage{url}
\usepackage{verbatim}
\usepackage{graphicx}
\usepackage{cite}
\usepackage{color}
\usepackage{amsthm}

\newtheorem{theorem}{Theorem}

\begin{document}

\title{Modulo Sampling in Shift-Invariant Spaces: Recovery and Stability Enhancement}

\author{Yhonatan Kvich, ~\IEEEmembership{Graduate Student Member,~IEEE,} Yonina C. Eldar,~\IEEEmembership{Fellow,~IEEE}
\thanks{This research was supported by the Tom and Mary Beck Center for Renewable Energy as part of the Institute for Environmental Sustainability (IES) at the Weizmann Institute of Science, by the European Research Council (ERC) under the European Union’s Horizon 2020 research and innovation program (grant No. 101000967) and by the Israel Science Foundation (grant No. 536/22).}}



\maketitle

\begin{abstract}
Sampling shift-invariant (SI) signals with a high dynamic range poses a notable challenge in the domain of analog-to-digital conversion (ADC). It is essential for the ADC's dynamic range to exceed that of the incoming analog signal to ensure no vital information is lost during the conversion process. Modulo sampling, an approach initially explored with bandlimited (BL) signals, offers a promising solution to overcome the constraints of dynamic range. In this paper, we expand on the recent advancements in modulo sampling to encompass a broader range of SI signals. Our proposed strategy incorporates analog preprocessing, including the use of a mixer and a low-pass filter (LPF), to transform the signal into a bandlimited one. This BL signal can be accurately reconstructed from its modulo samples if sampled at slightly above its Nyquist frequency. The recovery of the original SI signal from this BL representation is then achieved through suitable filtering. We also examine the efficacy of this system across various noise conditions. Careful choice of the mixer plays a pivotal role in enhancing the method's reliability, especially with generators prone to instability. Our approach thus broadens the  framework of modulo sampling's utility in efficiently recovering SI signals, pushing its boundaries beyond BL signals while sampling only slightly above the rate needed for a SI signal.
\end{abstract}

\begin{IEEEkeywords}
Modulo sampling, dynamic range, unlimited sampling, $B^2 R^2$.
\end{IEEEkeywords}

\section{Introduction}
\label{sec:intro}

Analog-to-digital converters (ADCs) transform analog signals into a digital format for processing in digital signal processing systems. The cost and power requirements of ADCs escalate with an increase in the sampling rate, making it preferable to operate at the minimum necessary rate for efficient sampling, as highlighted in previous studies \cite{eldar2015sampling,mishali2011sub}. The Shannon-Nyquist sampling theorem is commonly employed in this context, asserting that bandlimited (BL) signals can be accurately represented by their samples at a rate at least twice the maximum frequency present in the signal. As the sampling rate increases, so do the cost and power consumption associated with the analog-to-digital conversion process. Therefore, sampling as close as possible to the signal's Nyquist rate is beneficial.
Another important factor is the dynamic range of ADCs. To avoid signal clipping and the consequent loss of information, as illustrated in Fig. \ref{fig:compare}(a), the dynamic range of an ADC must surpass that of the input analog signal.

\begin{figure}[htb]
	\begin{minipage}[b]{\linewidth}
		\centering
		\centerline{\includegraphics[width=\columnwidth]{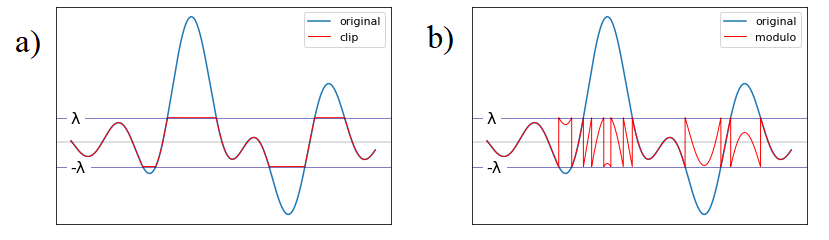}}
	\end{minipage}
	\caption{(a) BL signal and its clipped version. (b) BL signal and output of its modulo operation.}
	\label{fig:compare}
\end{figure}

Numerous approaches have been developed to address clipping or to enhance the dynamic range, with classification based on the presence or absence of preprocessing prior to sampling. Techniques that forgo preprocessing exploit the correlation inherent in BL signals' samples taken above the Nyquist rate, using this correlation to reconstruct any missing information due to clipping by oversampling \cite{marks1983restoring, marks1984error}. Alternatively, leveraging the spectral vacancies inherent in multiband communication systems can distinguish between original and clipped signals \cite{abel1991restoring, rietman2008clip}. These methods, however, either necessitate significant oversampling \cite{marks1983restoring, marks1984error} or presuppose knowledge of the spectral gaps \cite{abel1991restoring, rietman2008clip}, without offering theoretical assurances for their efficacy. Clipping can also be circumvented through attenuation, though this approach risks diminishing low-amplitude signals beneath the noise floor. Variable gain attenuators, like automatic gain controls (AGCs) and companders, adjust to preserve signal integrity without disproportionately affecting signals of smaller amplitude. AGCs utilize a series of amplifiers regulated by feedback to maintain a consistent output level \cite{perez2011automatic, mercy1981review}, while companding adjusts gain inversely proportional to the signal amplitude. Companding, however, like clipping, introduces nonlinear distortion and expands the signal's bandwidth. Beurling's theorem and subsequent methodologies by Landau et al. have facilitated the recovery of BL signals from their companded versions by ensuring the compander's output for finite energy inputs remains finite and the compander function is monotone and differentiable across the input signal's dynamic range \cite{landau1960recovery, landau1961recovery}. Despite these advancements, the implementation of companding-based solutions at minimal sampling rates faces challenges due to requirements for monotonicity, differentiability, and finite energy output \cite{landau1960recovery, landau1961recovery}.

An alternate strategy involves applying a modulo operation to the input signal before sampling to constrain its dynamic range, a technique that has found application in high-dynamic-range ADCs, or self-reset ADCs, in imaging contexts \cite{park2007wide, sasagawa2015implantable, yuan2009activity, krishna2019unlimited}. This approach, alongside storing modulo signal samples, often involves capturing additional data, such as the extent of folding for each sample or the folding direction, complicating the sampling circuitry while simplifying signal reconstruction from folded samples. See Fig. \ref{fig:compare}(b) for a visual representation.

The concept of unlimited sampling, introduced by Bhandari et al., relies solely on folded or modulo samples for signal recovery, demonstrating that sampling above the Nyquist rate enables unique identification of BL signals from their modulo samples \cite{bhandari2020unlimited}. This method, which extends Itoh’s unwrapping algorithm, shows that a sufficient oversampling rate allows for the computation of the original signal's higher-order differences from modulo samples, facilitating signal reconstruction through summation of these differences. However, this technique's effectiveness diminishes in noisy environments, necessitating a significantly higher oversampling rate for reliable recovery \cite{bhandari2020unlimited}. Subsequent improvements by Romanov and Ordentlich and others have sought to reduce the required sampling rate, apply the technique to various signal models, and explore hardware implementations for high-dynamic-range ADCs using modulo operations \cite{romanov2019above, lu2020high,bhandari2021unlimited, rudresh2018wavelet, bhandari2018unlimited, musa2018generalized, prasanna2020identifiability, ji2019folded, Bhandari_Krahmer_2020, Guo_Bhandari_2023}. Despite these advancements, challenges remain in terms of missing theoretical guarantees, stability concerns, the need for smooth and monotone operators, and reliance on higher-than-Nyquist sampling rates.

Azar et al. \cite{azar2022robust, Azar_Mulleti_Eldar_2022a} introduced a recovery algorithm capable of reconstructing BL signals from their modulo samples at a rate slightly above the Nyquist rate, even under various noise conditions. This approach was further refined in \cite{Shah_Mulleti_Eldar_2023} through the incorporation of a sparsity assumption to enhance noise robustness. Additionally, Mulleti et al. \cite{mulleti2022modulo} explored the application of modulo sampling to finite-rate-of-innovation (FRI) signals. The methodologies developed for both BL and FRI signals have been successfully implemented in hardware, as detailed in \cite{mulleti2023hardware}.


While the assumption of bandlimitation often serves as a practical approximation, numerous signals exhibit more precise representations through alternative bases or possess distinct structures within the Fourier domain
\cite{unser2000sampling,eldar2009compressed, bhandari2018unlimited,rudresh2018wavelet,prasanna2020identifiability,ji2022unlimited,mulleti2022modulo}.
Shift-invariant (SI) spaces, in particular, hold a crucial role in the theory of sampling. Signals within these spaces are expressed as linear combinations of shifts from a collection of generating functions \cite{eldar2015sampling, eldar2009compressed,deboor1994structure,christensen2004oblique,aldroubi2001nonuniform, bhandari2011shift}. Any such signal can be accurately reconstructed in an SI space, formed by shifting $k$ functions with period $T$, using $k$ distinct sampling sequences. These samples are derived by passing the signal through a filter bank consisting of $k$ filters and sampling the output uniformly at intervals of $nT$. This technique allows for a sampling rate of $k/T$, potentially much lower than the Nyquist rate for signals with wideband generators, given certain conditions \cite{eldar2009compressed}.

A unified theory for modulo sampling within SI spaces remains undeveloped. Bhandari and Krahmer \cite{Bhandari_Krahmer_2020} considered 2D B-splines, theirs methodology necessitates an oversampling rate determined by the modulo operation's requirements and the signal's amplitude limits, and lacks noise robustness.

This study introduces an innovative recovery approach for SI signals utilizing modulo sampling, which extends upon the methodology presented in \cite{kvich2024Modulo}. To enhance stability, we integrate a mixer into the signal processing framework. This integration is designed to exploit the full spectral range of the signal generator, significantly bolstering the stability of the recovery process. By initially applying a low-pass filter (LPF) to transform the input signal into a BL form suitable for modulo recovery, as described in prior work \cite{azar2022robust, Azar_Mulleti_Eldar_2022a}, and subsequently leveraging the characteristics of the SI space for signal reconstruction, our approach requires only a slightly higher sampling rate than that necessary for SI signals. The key improvement facilitated by the mixer not only enables the efficient utilization of the entire spectrum of the generator but also effectively addresses and mitigates potential stability concerns, resulting in a more resilient and efficient signal recovery mechanism.

We perform simulations with Lorentzian generators and achieve successful signal recovery amidst noise interference. Our results reveal enhanced noise robustness through the integration of the mixer. We also looked at scaled B-spline generators, which typically present difficulties in SI recovery due to inherent instability. Our observations confirm that the addition of a mixer substantially improves both stability and the accuracy of recovery in noisy environments, highlighting the effectiveness of our method in enhancing signal reconstruction for various generator models. Following this, we apply modulo sampling to Electrocardiogram (ECG) signals and successfully achieve recovery, demonstrating the applicability of our method to biological signals.
	
The structure of this paper is outlined as follows: Section \ref{sec:background} introduces the foundational concepts of modulo sampling, SI spaces, and the formulation of the problem. Section \ref{sec:system} details the design of the recovery system, including a theoretical framework for signal reconstruction. In Section \ref{sec:sim}, we illustrate the system's performance through simulations, particularly focusing on its behavior in noisy conditions. Section \ref{sec:ecg} explores the application of SI modulo sampling in the context of ECG signal analysis. Finally, Section \ref{sec:conclusion} offers concluding remarks.

Throughout this paper, the following notations are used. The space of sequences with finite norm is represented as $\ell^2$, and its corresponding norm is denoted $\|\cdot\|_2$. For a given sequence $a[n]$, its Discrete-Time Fourier Transform (DTFT) is $A(e^{j\omega}) := \sum_{n\in\mathbb{Z}} a[n] e^{-j\omega n}$. Similarly, for a function $x(t)$ with finite norm, its Continuous-Time Fourier Transform (CTFT) is defined as $X(\omega) := \int_{t\in\mathbb{R}} x(t) e^{-j\omega t}$,  we also employ the notation \(\mathcal{F}\{x\}(\omega)\). Denote $\ast$ as convolution.


\section{Background and Problem Statement}
\label{sec:background}

\subsection{Preliminaries}

The dynamic range of ADCs presents a significant challenge in accurately capturing signals. An issue that often arises is clipping, which occurs when the signal values exceed the ADC's dynamic range, resulting in a loss of information. One approach to circumvent this problem is to expand the dynamic range, but this leads to increased quantization errors and necessitates the use of higher-resolution ADCs, which in turn consume more power. An alternative proposed solution is the application of a nonlinear operator to the analog signal before sampling. This paper primarily focuses on the modulo operator as the non-linear operator of interest, diverging from the alternatives explored in \cite{azar2022robust, Azar_Mulleti_Eldar_2022a}. The modulo operator works by reducing the dynamic range of a signal, ``folding" it within a predetermined interval. For a given $\lambda>0$, the operator maps real values to the interval $[-\lambda, \lambda]$ as follows:
\begin{equation}
	\mathcal{M}_\lambda x := ((x+\lambda) \mod 2\lambda) - \lambda .
\end{equation}

Azar et al. \cite{azar2022robust, Azar_Mulleti_Eldar_2022a} have shown that by sampling slightly above the Nyquist rate, it is feasible to ``unfold" the samples of a BL signal. They also developed an algorithm for signal recovery in such contexts, called $B^2 R^2$.
Fig \ref{fig:azar} illustrates the overall framework of the modulo recovery process for BL signals. An input BL signal is first filtered through $s(t)$, processed by the analog modulo operator.
Finally, the signal is then sampled at a rate $T_s < T$. Here, $T$ denotes the Nyquist rate, which is defined as twice the highest frequency found within the bandwidth of $x(t)$.
Various other methods for recovering BL signals have been proposed, such as the one in \cite{Shah_Mulleti_Eldar_2023}, which utilizes sparsity assumptions and fast-ISTA (FISTA) to improve the method in \cite{azar2022robust, Azar_Mulleti_Eldar_2022a}. Techniques in \cite{bhandari2020unlimited,romanov2019above} and \cite{Guo_Bhandari_2023} have introduced methodologies that use higher-order differences and iterative signal sieving, respectively, for signal reconstruction. However these approaches require higher-than-Nyquist sampling rates and are less robust to noise.

\begin{figure}[htb]
	\begin{minipage}[b]{\linewidth}
		\centering
		\centerline{\includegraphics[width=\columnwidth]{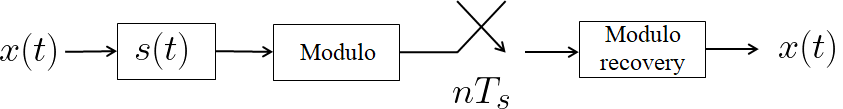}}
	\end{minipage}
	\caption{The modulo recovery system for BL signals, as introduced in \cite{azar2022robust, Azar_Mulleti_Eldar_2022a, Shah_Mulleti_Eldar_2023}.}
	\label{fig:azar}
\end{figure}

This study aims to expand the findings of \cite{azar2022robust, Azar_Mulleti_Eldar_2022a, Shah_Mulleti_Eldar_2023} to the domain of SI signals \cite{eldar2015sampling}. A SI space with a single generator is defined as the span of a generator signal $h(t)$ with time-shifts of $T$, meaning
\begin{equation}
	\mathcal{H} := \Big\{ \sum_{n\in\mathbb{Z}}{a[n]h(t-n T)} | a\in\ell^2 \Big\} .	\label{eq:si}
\end{equation}
Choosing $h(t)$ as the sinc function with a width of $\frac{1}{2T}$, the space in (\ref{eq:si}) comprises $2\Omega$-BL signals, where $\Omega=\frac{1}{2T}$.

It is a well-established fact \cite{eldar2015sampling} that signals in this space can be successfully recovered from samples taken at a rate of $\frac{1}{T}$, independent of the Nyquist rate of the generator or whether it is BL. The recovery algorithm involves frequency domain division using a correction filter, the design of which is contingent upon the properties of the SI space, the sampling rate, and the analog filter applied before sampling.

The domain of SI spaces offers increased flexibility in terms of signal sampling, leading to a broader range of potential applications. The objective here is to initiate a comprehensive study on modulo sampling of SI signals, aiming for robust recovery while sampling at a rate close to $\frac{1}{T}$.

\subsection{Problem Formulation and Recovery Approach}

Our goal is to reconstruct a SI signal
\begin{equation}
\label{eq:si_signal}
	x(t)=\sum_{n\in\mathbb{Z}}{a[n]h(t-n T)}\in\mathcal{H}
\end{equation}
from its modulo samples $\{\mathcal{M}_\lambda x [n T_s]\}$, where $f_s=\frac{1}{T_s}$ is the sampling frequency, at least \(\frac{1}{T}\), nearing its value. Notably, \(T\) can be significantly lower than the Nyquist rate of both \(h(t)\) and \(x(t)\). To facilitate this, we first apply analog pre-processing before implementing the modulo operator. Our objective is to reconstruct the input signal from the modulo samples $\{\mathcal{M}_\lambda f(x) [n T_s]\}$, with $f(\cdot)$ denoting the necessary analog domain pre-processing. These pre-processing steps are critical for ensuring unique recovery and augmenting the efficiency of the modulo sampling method. We will demonstrate that modulo signal recovery is achievable for \(T_s < T\), for a broad class of generators $h(t)$.


\section{Recovery System and Theoretical Guarantees}
\label{sec:system}

\subsection{Recovery of SI Signal}

In this section, we introduce the recovery system, its algorithm, and a theorem that validates the uniqueness of recovery in SI space, as defined in equation (\ref{eq:si_signal}). We explain the system mechanics and mathematical principles, followed by a discussion on method stability, highlighting its effectiveness and areas for potential improvement.

\begin{figure}[htb]
	\begin{minipage}[b]{\linewidth}
		\centering
		\centerline{\includegraphics[width=\columnwidth]{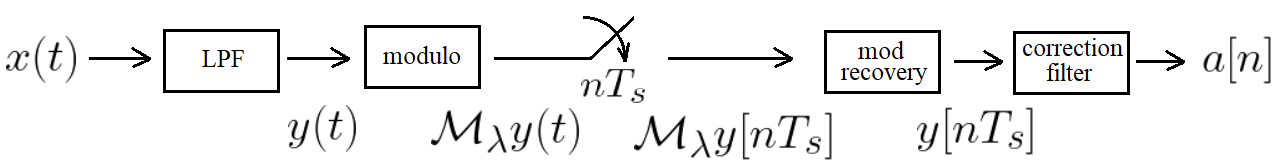}}
	\end{minipage}
	\caption{Block diagram for modulo sampling for SI spaces as detailed in \cite{kvich2024Modulo}. See theorem \ref{thm:recovery_icassp} for more details on the system and proof of recovery.}
	\label{fig:kvich_icassp}
\end{figure}

\begin{theorem}
	\label{thm:recovery_icassp}
	
	Let \(\mathcal{H}\) be a SI space as in (\ref{eq:si}) and let \(x(t) \in \mathcal{H}\) be the input signal. LPF with a cutoff frequency of \(\frac{\pi}{T}\) is applied to yield \(y(t) = \text{LPF}(x(t))\). Modulo operation is applied to \(y(t)\), resulting in \(\mathcal{M}_\lambda y(t)\). The signal \(\mathcal{M}_\lambda y(t)\) is then sampled at rate \(T_s < T\), producing the samples \(\mathcal{M}_\lambda y[n T_s]\).
	
	The input signal \(x(t)\) can be perfectly recovered from the samples \(\mathcal{M}_\lambda y[n T_s]\), provided that \(H(\omega)\) is non-zero almost everywhere within the frequency interval \([-\frac{\pi}{T}, \frac{\pi}{T}]\).
\end{theorem}

\begin{proof}
	
	We can see that $y(t)$ is BL with Nyquist rate of $T$ and sampled at rate \(T_s < T\). Following the approach in \cite{azar2022robust, Azar_Mulleti_Eldar_2022a}, this allows for the reconstruction of the signal $y(t)$, meaning we obtain $Y(\omega)$ for $\omega$ within $[-\frac{\pi}{T}, \frac{\pi}{T}]$.
	
	Since $y(t)=\text{LPF}(x(t))$, we know that $Y(\omega) = X(\omega)$ for $\omega\in[-\frac{\pi}{T}, \frac{\pi}{T}]$, from the definition of LPF.
	Following equation (\ref{eq:si_signal}), we can apply a Fourier transform and using the properties of DTFT and CTFT as shown in \cite{eldar2015sampling}, obtaining $X(\omega) = A(e^{j\omega T}) H(\omega)$. So we conclude that $Y(\omega) = A(e^{j\omega T}) H(\omega)$ for $\omega\in[-\frac{\pi}{T}, \frac{\pi}{T}]$.

	Given the theorem's assumption, $A(e^{j\omega T})$ can be computed as $\frac{Y(\omega)}{H(\omega)}$ for $\omega\in[-\frac{\pi}{T}, \frac{\pi}{T}]$.
	A discrete-time filter with a frequency response of $\frac{1}{H(\omega)}$ within the frequency range $|\omega|\le \frac{\pi}{T}$ is applied.
	This filter acts as a correction filter, which ultimately leads to the extraction of $a[n]$ and thereby the reconstruction of the original input signal $x(t)$ using (\ref{eq:si_signal}).
\end{proof}

The crucial assumption in theorem \ref{thm:recovery_icassp} is that \(H(\omega) \ne 0\) almost everywhere within this frequency range, a condition met by most generators used in practice. If $H(\omega) = 0$ for any $\omega$ within the range $[-\frac{\pi}{T}, \frac{\pi}{T}]$, then the correction filter will enhance noise at those frequencies.


To demonstrate stability within our recovery system, we explore the example of a stable generator, specifically B-spline and Lorentzian.
Define $\beta^{(0)}(t) = \mathbf{1}_{[-0.5, 0.5]}(t)$ as the $0$-th order B-spline \cite{eldar2015sampling}. The $n$-th order B-spline is recursively defined as $\beta^{(n+1)} = \beta^{(n)} \ast \beta^{(0)}$. The CTFT of an $n$-th order spline is $\text{sinc}^{n+1} (\omega)$. With no roots within $[-\frac{\pi}{T}, \frac{\pi}{T}]$, the correction filter becomes stable, with frequency response of $\frac{1}{\text{sinc}^{n+1} (\omega)}$ for $\omega \in [-\frac{\pi}{T}, \frac{\pi}{T}]$.

A second example for stable generator is Lorentzian function, defined as

\begin{equation}
	\label{eq:lorz}
	h_\gamma(t)= \frac{1}{\pi \gamma \Big( 1+ (\frac{x}{\gamma})^2 \Big)}
\end{equation}
where $\gamma>0$ is a scale parameter. The CTFT is

\begin{equation}
	\label{eq:lorz_ft}
	H_\gamma(\omega)=\exp(-\gamma |\omega|).
\end{equation}

The correction filter have an impulse response given by $\frac{1}{H_\gamma(\omega)} = \exp(\gamma |\omega|)$ for $\omega \in [-\frac{\pi}{T}, \frac{\pi}{T}]$, thereby robust to noise. Fig. \ref{fig:generators} displays the Fourier transform of Lorentzian generators with $\gamma=0.25$ and $\gamma=0.5$.
Fig. \ref{fig:snr} displays our simulation results for Lorentzian generators, which yield a stable correction filter. Notably, the final application of the correction filter does not significantly alter the mean squared error (MSE). We will delve into more details in Section \ref{sec:sim}.

However, some generators demonstrate instability. This is particularly evident with scaled first-order splines, as discussed in section \ref{sec:sim}. Fig. \ref{fig:R} illustrates the instability of the correction filter, attributed to roots at the central frequency in $H(\omega)$. This root presence amplifies the recovery error in the final step, as illustrated in Fig. \ref{fig:snr_2spline}, where the correction filter exacerbates the error by approximately 30 dB across the varying noise levels. This highlights the necessity for a more robust approach to address this instability, which we introduce in the subsequent section.

\subsection{Improving Robustness}

In this section, we address the enhancement of robustness in the presence of noise. During the recovery process discussed previously, we measure $\mathcal{M}_\lambda y[n T_s]$. When noise is present, we assume the measurement is $\mathcal{M}_\lambda y[n T_s] + e[n]$, where $e[n]$ represents the noise. We begin with a BL modulo recovery technique as detailed in \cite{azar2022robust, Azar_Mulleti_Eldar_2022a}, resulting in $\mathcal{M}_\lambda y[n T_s] + e[n] \to \hat{y}[n T_s]$. Naturally, there is an error $e_{\text{BL}}[n] = \hat{y}[n T_s] - y[n T_s]$. Note that the modulo recovery algorithm is nonlinear, and a comprehensive analysis of its robustness is detailed in the cited works.

The subsequent phase, absent in the BL signal modulo recovery, involves division in the frequency domain by $H(\omega)$ for $\omega \in [-\frac{\pi}{T}, \frac{\pi}{T}]$. Let $\hat{a}[n]$ be the estimated coefficients. From the recovery process, we know from Theorem \ref{thm:recovery_icassp} that $A(e^{j\omega}) = \frac{Y(e^{j\omega})}{H(\omega)}$ and from the proposed recovery process $\hat{A}(e^{j\omega}) = \frac{\hat{Y}(e^{j\omega})}{H(\omega)}$ for $\omega \in [-\frac{\pi}{T}, \frac{\pi}{T}]$. The final reconstruction error is $e_{\text{coef}}[n] = \hat{a}[n] - a[n]$. We can analyze the error's DTFT as follows:

\begin{equation}
	\begin{split}
		E_{\text{coef}}(e^{j\omega}) = \hat{A}(e^{j\omega}) - A(e^{j\omega}) = \frac{\hat{Y}(e^{j\omega})}{H(\omega)} - \frac{Y(e^{j\omega})}{H(\omega)} =\\= \frac{1}{H(\omega)} E_{\text{BL}}(e^{j\omega})
	\end{split}
\end{equation}
for $\omega \in [-\frac{\pi}{T}, \frac{\pi}{T}]$. According to Parseval's theorem, the MSE in discrete time is equivalent to the MSE in the frequency domain, so:

\begin{equation}
	\label{eq:e_mse}
	\| e_{\text{coef}} \|_2^2 = \Big\| \frac{1}{H(\omega)} E_{\text{BL}}(e^{j\omega}) \mathbf{1}_{[-\frac{\pi}{T}, \frac{\pi}{T}]} \Big\|_2^2 .
\end{equation}

Many common generators distribute their energy unevenly across the frequency interval $[-\frac{\pi}{T}, \frac{\pi}{T}]$, leading to high MSE in equation (\ref{eq:e_mse}). The main concern arises when $H(\omega)$ has roots within $[-\frac{\pi}{T}, \frac{\pi}{T}]$, causing the error from the BL modulo recovery $e_{\text{BL}}$ to amplify the coefficients error $e_{\text{coef}}$, as shown in equation (\ref{eq:e_mse}).

To enhance the method's stability for such generators, we incorporate a mixer using a $T$-periodic function \(p(t)\) prior to the LPF. This transforms the input signal into a BL signal with a Nyquist rate of \(T\) in the analog domain, resulting in \(y(t) = \text{LPF}\Big(p(t) x(t)\Big)\). The subsequent steps mirror the recovery process described in Theorem \ref{thm:recovery_icassp}, including the application of an analog modulo operator, sampling above the Nyquist rate, and using a proposed modulo recovery algorithm \cite{azar2022robust, Azar_Mulleti_Eldar_2022a, Shah_Mulleti_Eldar_2023}. The final step applies a correction filter different from $\frac{1}{H(\omega)}$ that will be specified. The mixer enables the use of a broader frequency range of $H(\omega)$ to create a favorable correction filter that minimizes noise amplification as seen in equation (\ref{eq:e_mse}). This avoids singularities and results in a more constant filter across frequencies. The mixer selection is tailored to the specific needs of the SI space, sampling rate, and noise levels.

The added mixer operates on the analog signal, allowing for the use of a wider frequency range of $H(\omega)$ without necessitating a higher sampling rate. The only difference in the recovery process is the use of a different digital correction filter, which is known derives from system's setting.
Figure \ref{fig:R} provides a visual representation, showing the correction filter for a scaled B-spline with two singularities, indicating unstable reconstruction. It also illustrates the correction filter with an appropriately selected mixer, which avoids the singularities and noise amplification as seen in equation (\ref{eq:e_mse}), leading to improved reconstruction.

\begin{figure*}[htb]
		\centering
		\centerline{\includegraphics[width=14cm]{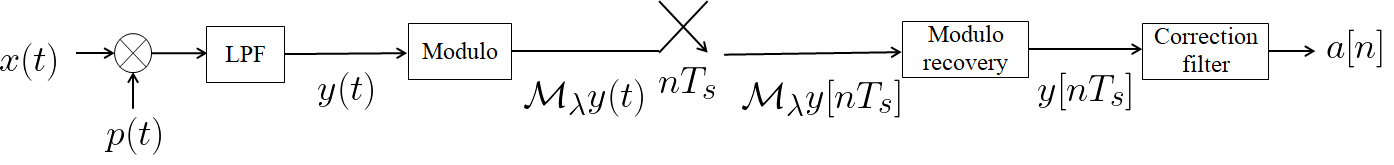}}
	\caption{Block diagram for modulo sampling for SI spaces with improved stability. See theorem \ref{thm:recovery} for more details on the system and proof of recovery.}
	\label{fig:sys}
\end{figure*}

Let \(\mathcal{H}\) represent a SI space as defined in (\ref{eq:si}). Consider a $T$-periodic signal $p(t)$. It can be expressed using its Fourier series representation as $p(t) = \sum_{l\in\mathbb{Z}}{c_l e^{-j\frac{2\pi}{T} l t}}$, where the coefficients $c_l$ are given by $c_l = \frac{1}{T} \int_{-\frac{T}{2}}^{\frac{T}{2}}{p(t) e^{j\frac{2\pi}{T} l t}} dt$. Denote

\begin{equation}\label{eq:R}
	R(\omega) =  \sum_{l\in\mathbb{Z}}{c_l  H\Big(\omega - \frac{2\pi l}{T}\Big)}.
\end{equation}

\begin{theorem}
	\label{thm:recovery}
	Let $\mathcal{H}$ be a SI space as in (\ref{eq:si}) and the input signal $x(t)\in\mathcal{H}$.
	A mixer with $T$-periodic signal $p(t)$ and LPF with cut-off frequency $\frac{\pi}{T}$ are applied, resulting in $y(t)= \text{LPF}\Big( p(t)x(t) \Big)$. The modulo operator is applied, giving $\mathcal{M}_\lambda y(t)$. The sampling rate is $T_s < T$, resulting in the samples $\mathcal{M}_\lambda y[n T_s]$.
	
	The input signal $x(t)$ can be perfectly recovered from the samples $\mathcal{M}_\lambda y[n T_s]$, assuming that $R(\omega)$, as defined in (\ref{eq:R}), is non-zero almost everywhere for $\omega\in[-\frac{\pi}{T}, \frac{\pi}{T}]$.

\end{theorem}

\begin{proof}
	From the properties of CTFT and DTFT \cite{eldar2015sampling} we can tell that 
	\begin{equation}
			\mathcal{F}\Big\{ p(t)x(t) \Big\} = \sum_{l\in\mathbb{Z}}{c_l X\Big(\omega - \frac{2\pi l}{T}\Big)}
	\end{equation}
additionally from (\ref{eq:si_signal}) and \cite{eldar2015sampling}, we can see that $X(\omega) = A(e^{j\omega T}) H(\omega)$. We conclude that
	
	\begin{equation}
		\label{eq:px_ft}
		\begin{split}
		\mathcal{F}\Big\{ p(t)x(t) \Big\} = \sum_{l\in\mathbb{Z}}{c_l A(e^{j\omega T}) H\Big(\omega - \frac{2\pi l}{T}\Big)} =\\ A(e^{j\omega T}) \Bigg( \sum_{l\in\mathbb{Z}}{c_l H\Big(\omega - \frac{2\pi l}{T}\Big)} \Bigg) = A(e^{j\omega T}) R(\omega).
		\end{split}
	\end{equation}
	After applying LPF on $p(t) x(t)$, we get a BL signal $y(t)$. From the definition of LPF we know that $\mathcal{F}\Big\{ p(t)x(t) \Big\} (\omega) = Y(\omega)$ for $\omega\in[-\frac{\pi}{T}, \frac{\pi}{T}]$. Thus, with the use of (\ref{eq:px_ft})
	
	\begin{equation}\label{eq:y_ft}
		Y(\omega) = A(e^{j\omega T}) R(\omega), \quad\omega\in\Big[-\frac{\pi}{T}, \frac{\pi}{T}\Big].
	\end{equation}
Note that $y(t)$ is BL with Nyquist rate of $T$, modulo operator is applied and than sampled at rate $T_s < T$, resulting in the samples $\mathcal{M}_\lambda y[n T_s]$. Following \cite{azar2022robust, Azar_Mulleti_Eldar_2022a}, we can recover the signal $y(y)$, meaning we get $Y(\omega)$ for $\omega\in[-\frac{\pi}{T}, \frac{\pi}{T}]$.
	
	Using equation (\ref{eq:y_ft}), we can compute $A(e^{j\omega T}) = \frac{Y(\omega)}{R(\omega)}$. Note that we required in the theorem that $R(\omega)\ne 0$ almost everywhere for $\omega$ in the interval $[-\frac{\pi}{T}, \frac{\pi}{T}]$.
	We apply a discrete-time filter with a frequency response of $\frac{1}{R(\omega)}$ over the interval $\omega \in [-\frac{\pi}{T}, \frac{\pi}{T}]$.
	This operator serves as the correction filter. Thereby obtaining $a[n]$ and consequently recovering the input signal $x(t)$ using (\ref{eq:si_signal}).

\end{proof}

Note that if we select $p(t) = 1$ than $R(\omega)= H(\omega)$. From the assumption in the theorem, we can see that the correction filter $\frac{1}{R(\omega)} = \frac{1}{H(\omega)}$ in properly defined.
Fig. \ref{fig:pipeline} illustrates a noise-free recovery scenario, showcasing the complete sequence of system operations. The input signal resides within the SI space formed by a Lorentzian function, as defined in (\ref{eq:lorz}), having a scale parameter $\gamma=0.5$, and the space's shift is defined as $T=1$. The SI space is $\mathcal{H}=\Big\{ \sum_{n\in\mathbb{Z}}{a[n]h_{0.5}(t-n)} | a\in\ell^2 \Big\}$. Note that here we do not use a mixer, we set $\lambda=0.1$ and the oversampling rate is set to $5$.
In Fig. \ref{fig:pipeline}(a), the progression of the signal, its state after undergoing LPF, and the outcome after applying the modulo operation is depicted. 
Note that the original signal is not BL. The application of a LPF, however, enables subsequent modulo recovery.
Fig. \ref{fig:pipeline}(b) displays both the true parameters $a[n]$ and the parameters recovered through the process.

\begin{figure*}[htb]
	
	\centering
	\centerline{\includegraphics[width=\textwidth]{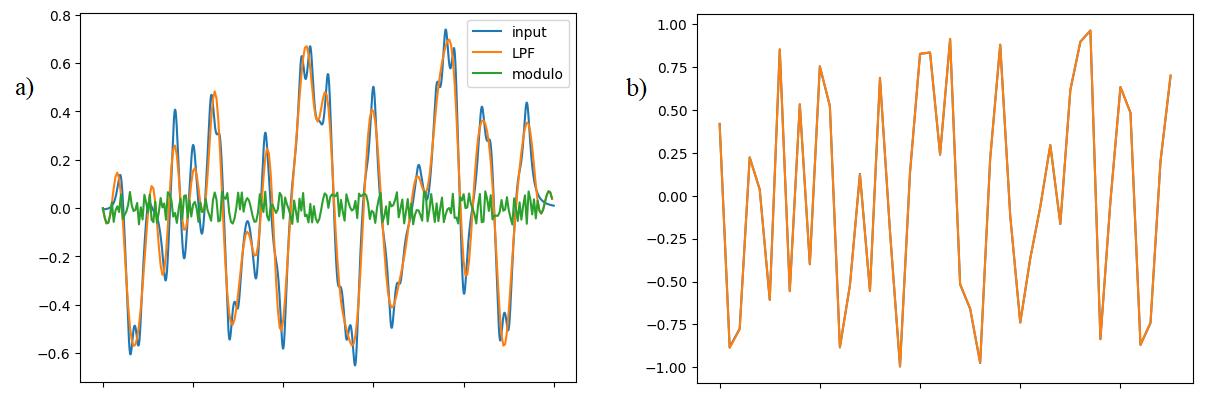}}

	\caption{Signal recovery for a Lorentzian-based signal within the SI space. (a) the initial input signal (blue) $x(t)= \sum_{n\in\mathbb{Z}}{a[n] h_{0.5}(t-n)} \in \mathcal{H}$, followed by LPF, $y(t)=\text{LPF}(x(t))$ (orange), and modulo operation $\mathcal{M}_\lambda(y(t))$ (green). Subfigure (b) the parameter recovery phase, meaning $a[n]$ and the reconstruction $\hat{a}[n]$.}
	\label{fig:pipeline}
\end{figure*}

\section{Simulations and Noise Resilience}
\label{sec:sim}

This section illustrates the modulo sampling process within SI spaces, particularly focusing on the effects of noise. A key issue is the influence of the LPF on system energy loss.
We utilize a Lorentzian function generator $h_\gamma(t)$, as defined in (\ref{eq:lorz}) with scales $\gamma=0.5$ and $\gamma=0.25$, the Fourier transform is shown in (\ref{eq:lorz_ft}). The SI space, as in (\ref{eq:si_signal}) will use $T=1$. The study compares two scenarios: one without a mixer and another with a non-trivial mixer, which exploits the generator's higher frequencies to stabilize the correction filter.

Our simulations involved 500 distinct signals, each generated with 50 coefficients $a[n]$ drawn from a uniform distribution over $[-1, 1]$. For the scenario incorporating a mixer, the signals $x(t)$ were combined with $p(t)$, followed by LPF application with a cutoff frequency $\frac{\pi}{T}$, resulting in $y(t)=\text{LPF}(p(t) x(t))$. The mixer design was tailored to the SI space, while in the non-mixer scenario (effectively a constant mixer), LPF was directly applied. Post-normalization ensured the maximum absolute amplitude of the samples was one, before applying the modulo operator with $\lambda=0.2$ to obtain $\mathcal{M}_\lambda y(t)$. Sampling was then conducted at an oversampling rate of 5, denoted as \(T_s = \frac{T}{5}\), considering the Nyquist rate of the post-LPF signal $y(t)$ is $T$. This setup provided samples \(\mathcal{M}_\lambda y[n T_s]\). We introduced noise at various Signal-to-Noise Ratio (SNR) levels and utilized the $B^2 R^2$ algorithm \cite{azar2022robust, Azar_Mulleti_Eldar_2022a} for the modulo recovery of the BL signal $y[n T_s]$. To deduce the coefficients \(a[n]\), the correction filter from theorem \ref{thm:recovery} was applied, converting \(Y(\omega) \to\frac{Y(\omega)}{R(\omega)} = A(e^{j\omega T})\) for \(\omega \in [-\frac{\pi}{T},\frac{\pi}{T}]\), with $R(\omega)$ adjusted based on the mixer selection.

The approximated coefficients were denoted as \(\hat{a}[n]\). We evaluated MSE metrics, calculating \(\frac{\|a - \hat{a}\|_2^2}{\|a\|_2^2}\) for coefficient recovery and \(\frac{\|y[n T_s] - \hat{y}[n T_s]\|_2^2}{\|y[n T_s]\|_2^2}\) for modulo recovery, with \(\hat{y}[n T_s]\) indicating the reconstructed modulo samples. This method facilitated an independent assessment of modulo recovery stability.

Fig. \ref{fig:snr} graphically presents the results at various SNR levels, setting $p(t) = 1 + 200 \cos(\frac{2\pi}{T}t) + 200 \cos(\frac{4\pi}{T}t)$.
Analysis shows similar levels of MSE in both the modulo recovery process and the final recovery of coefficients, regardless of the mixer's use. The consistency in results is remarkable, considering the different approaches taken in generating the signals and the use of distinct correction filters in each instance, where the filters are expected to be stable. At lower SNR levels, the results are closely matched, with the mixer scenario showing a marginal improvement. However, at higher SNR levels, the performance without a mixer is notably better.

The LPF application incurs energy loss, impacting the system's performance while not affecting the modulo recovery step directly. For $\gamma=0.5$, the energy loss was about 2.2\%, in contrast to $\gamma=0.25$, which experienced a higher energy loss of 11\%. Fig. \ref{fig:generators} displays the CTFT of the generators and their lower energy bounds.

\begin{figure}[htb]
	\begin{minipage}[b]{\linewidth}
		\centering
		\includegraphics[width=7cm]{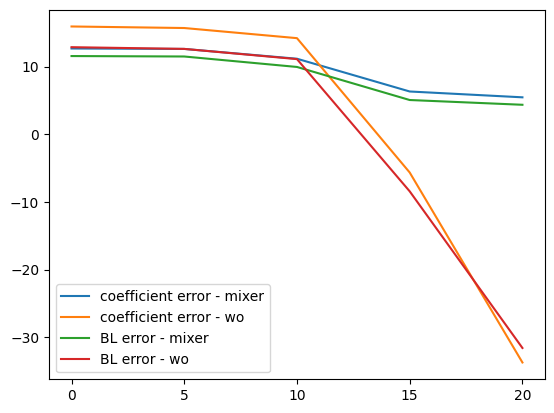}
		\includegraphics[width=7cm]{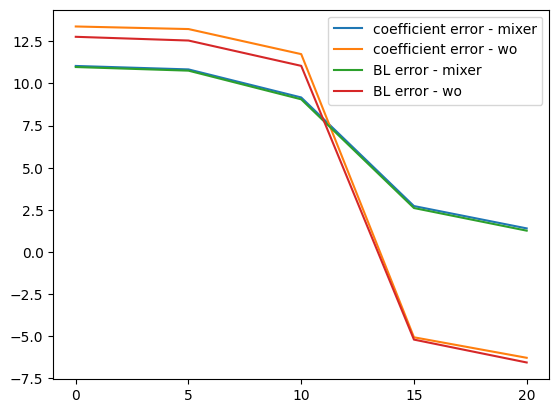}
	\end{minipage}
	\caption{MSE in dB for varying SNR values using a Lorentzian function generator with scales $\gamma=0.5$ (top) and $\gamma=0.25$ (bottom). ``BL error" refers to the recovery discrepancy of \(\hat{y}[n T_s]\), and ``coefficient error" relates to the precision in reconstructing \(\hat{a}[n]\), showcased for scenarios with and without mixer usage.}
	\label{fig:snr}
\end{figure}

Additional simulations on an alternative generator aimed to analyze scenarios with an unstable correction filter. Denote the first order B-spline as

\begin{equation}
	\beta^1(t) =
	\begin{cases}
		1 - |t|, \quad -1\le |t|\le 1\\
		0,\quad \text{o.w}
	\end{cases}.
\end{equation}

Employing a scaled first-order B-spline by $2.5$ for a generator $h(t)=\beta^1(\frac{t}{2.5})$ created a scenario where $H(\omega) = \text{sinc}^2 (2.5 \omega)$, with singularities within $[-\frac{\pi}{T},\frac{\pi}{T}]$. $\lambda$ remains he same but the mixer is set to $p(t) = 1 + 1000 \cos(\frac{2\pi}{T}t) + 1000 \cos(\frac{4\pi}{T}t)$ with Fig. \ref{fig:snr_2spline} showcasing the outcomes.
Using this particular generator, a significant decline in MSE is observed transitioning from the modulo recovery stage to the final coefficient error, a consequence of the correction filter's instability. The introduction of a mixer addresses this problem by producing a considerably more stable correction filter, thereby enhancing overall performance. Notably, while the BL recovery shows improved results in the absence of a mixer at higher SNR levels, the instability of the filter leads to a substantial worsening of the final MSE compared to that achieved with the mixer's application.

\begin{figure}[htb]
	\begin{minipage}[b]{\linewidth}
		\centering
		\includegraphics[width=7cm]{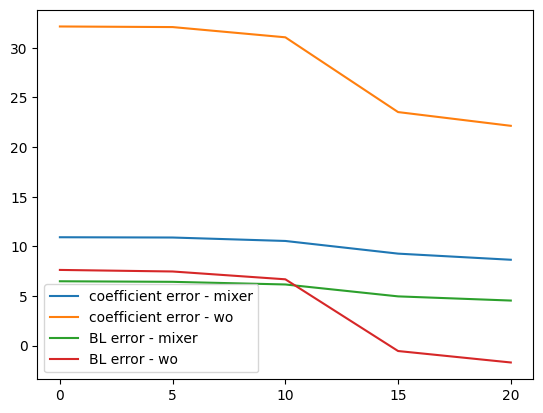}
	\end{minipage}
	\caption{MSE in decibels across SNR values for a scaled first-order B-spline generator. Error metrics follow the same definition as in Fig. \ref{fig:snr}, with results indicating the mixer's contribution to improved recovery outcomes at lower SNR levels.}
	\label{fig:snr_2spline}
\end{figure}

Fig. \ref{fig:R} displays the $R(\omega)$ function for both scenarios in the simulations, following $\ell_2$ normalization. In the no-mixer situation, the presence of roots is evident, whereas they are absent in the mixer case. Fig. \ref{fig:R} illustrates the correction filter, represented as $\frac{1}{R(\omega)}$. In the no-mixer scenario, significant peaks are observable, indicating that noise at those frequencies will be significantly amplified. Conversely, the correction filter in the mixer-utilized scenario exhibits much greater stability.

\begin{figure}[htb]
	\begin{minipage}[b]{\linewidth}
		\centering
		\includegraphics[width=\columnwidth]{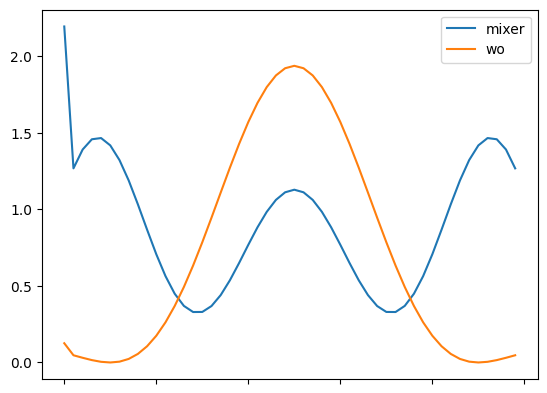}
		\includegraphics[width=\columnwidth]{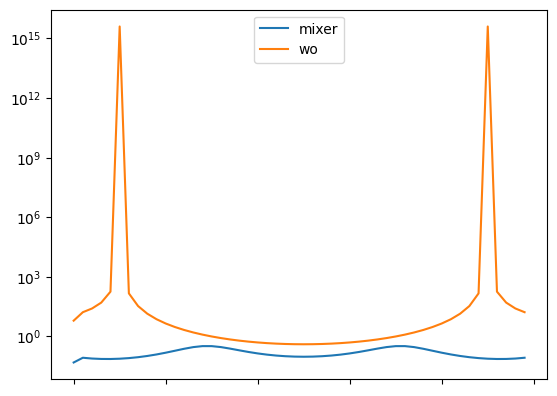}
	\end{minipage}
	\caption{On the top, the $R(\omega)$ function as determined in the simulations for both the mixer and non-mixer scenarios, post-$\ell_2$ normalization. On the bottom, the associated correction filter, $\frac{1}{R(\omega)}$, note the peaks observed in the non-mixer case.}
	\label{fig:R}
\end{figure}

The primary concern in modulo recovery $\mathcal{M}_\lambda f(x) [n T_s] \to f(x) [n T_s]$ varies with mixer use, as seen in the scaled-spline scenario. The correction filter's step $Y(\omega) \to \frac{Y(\omega)}{R(\omega)}$ emphasizes the importance of $h(t)$'s energy distribution across frequencies, as regions with diminished energy magnify noise through the filter.

\begin{figure}[htb]
	\begin{minipage}[b]{\linewidth}
		\centering
		\centerline{\includegraphics[width=8cm]{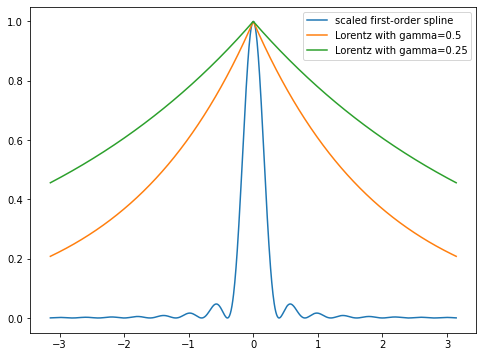}}
	\end{minipage}
	\caption{Fourier transforms of three generators within the interval $[-\frac{\pi}{T}, \frac{\pi}{T}]$. Scaled first-order spline (blue), Lorentzian with $\gamma=0.5$ (orange), and with $\gamma=0.25$ (green).}
	\label{fig:generators}
\end{figure}

The emphasis is on the impact of the LPF and mixer on the system's performance and the detailed examination of energy loss's role in the overall efficiency of the signal recovery process.


\section{Exploring Modulo Sampling in ECG Signal Analysis}
\label{sec:ecg}

ECG signals are vital for monitoring the heart's electrical activity, offering insights into cardiac health and identifying various heart conditions. These signals capture the sequential heartbeats, each comprising distinct components like the P wave, QRS complex, and T wave, which reflect different phases of the heart cycle \cite{Ajdaraga_Gusev_2017, Kwon_Jeong_Kim_Kwon_Park_Kim_Choi_2018}.

ECG signals exhibit a wide dynamic range due to the amplitude variations in heartbeat waveforms, stemming from physiological differences and the heart's activity level. The broad dynamic range poses challenges in signal acquisition and processing, where capturing the full scale of signal variations without distortion or loss of detail is critical. Modulo sampling presents a promising approach to address these challenges by enabling the effective compression of ECG signals' dynamic range, facilitating high-fidelity capture and reconstruction of the original waveform without the need for high-resolution ADCs.

For this study, we utilized data from \cite{Schellenberger2020}, a public dataset comprising ECG recordings from 30 healthy adults under various measurement conditions. This dataset, primarily focused on exploring radar-based vital sign detection, includes high-quality ECG signals sampled at 1000 Hz. Given that the spectrum of ECG signals is mostly contained up to 50 Hz, they can effectively be sampled at lower rates, such as 250 Hz or 500 Hz, without significant loss of information. However, the availability of ECG data sampled at 1000 Hz is advantageous, providing a high sampling rate that ensures comprehensive capture of the signal's details, making it particularly suitable for detailed analysis and processing tasks.

From this dataset, we extracted a single heartbeat from the first ECG channel of the first participant to serve as our base pulse. This pulse represents a typical cardiac cycle, embodying the characteristic waveform components essential for accurate ECG analysis.

Mathematically, the SI space will be modeled using the base pulse represented by $h(t)$, with a defined time shift $T>0$. This involves summing instances of the base pulse, each shifted by integer multiples of $T$. We will specify $T = 0.05$ seconds, corresponding to a frequency $f = \frac{1}{T} = 20$ Hz. The coefficients will comprise a sparse array where selected indices are set to one—representing the occurrences of heartbeats—and the remaining values are zero, thus indicating the precise moments of heartbeats.

Applying the modulo operator to this ECG signal model allows us to compress its dynamic range, facilitating the recording and analysis of ECG signals under limited resolution conditions. Subsequently, we employ a recovery algorithm designed to reconstruct the original ECG waveform from its modulo-sampled version. Through this process, we aim to accurately recover precise timing of each heartbeat, from the modulo samples. This methodology underscores the potential of modulo sampling in enhancing ECG signal analysis, promising improvements in both diagnostic tools and wearable health monitoring devices.

We apply a LPF to the SI signal, setting the cutoff frequency at $\frac{5\pi}{T}$. Following this, we implement the modulo operator, selecting $\lambda$ to be $0.1$ times the maximum amplitude of the signal. The signal is then sampled at an oversampling rate of 5. For the recovery of the modulo operation, we employ the methodology developed by Bahandari et al. \cite{bhandari2020unlimited}. The final step involves applying a correction filter, which is formulated based on the base pulse and the time shift $T$. Refer to Fig. \ref{fig:ecg} for a visual representation of the process.

\begin{figure*}[htb]
	\begin{minipage}[b]{\linewidth}
		\centering
		\includegraphics[width=\columnwidth]{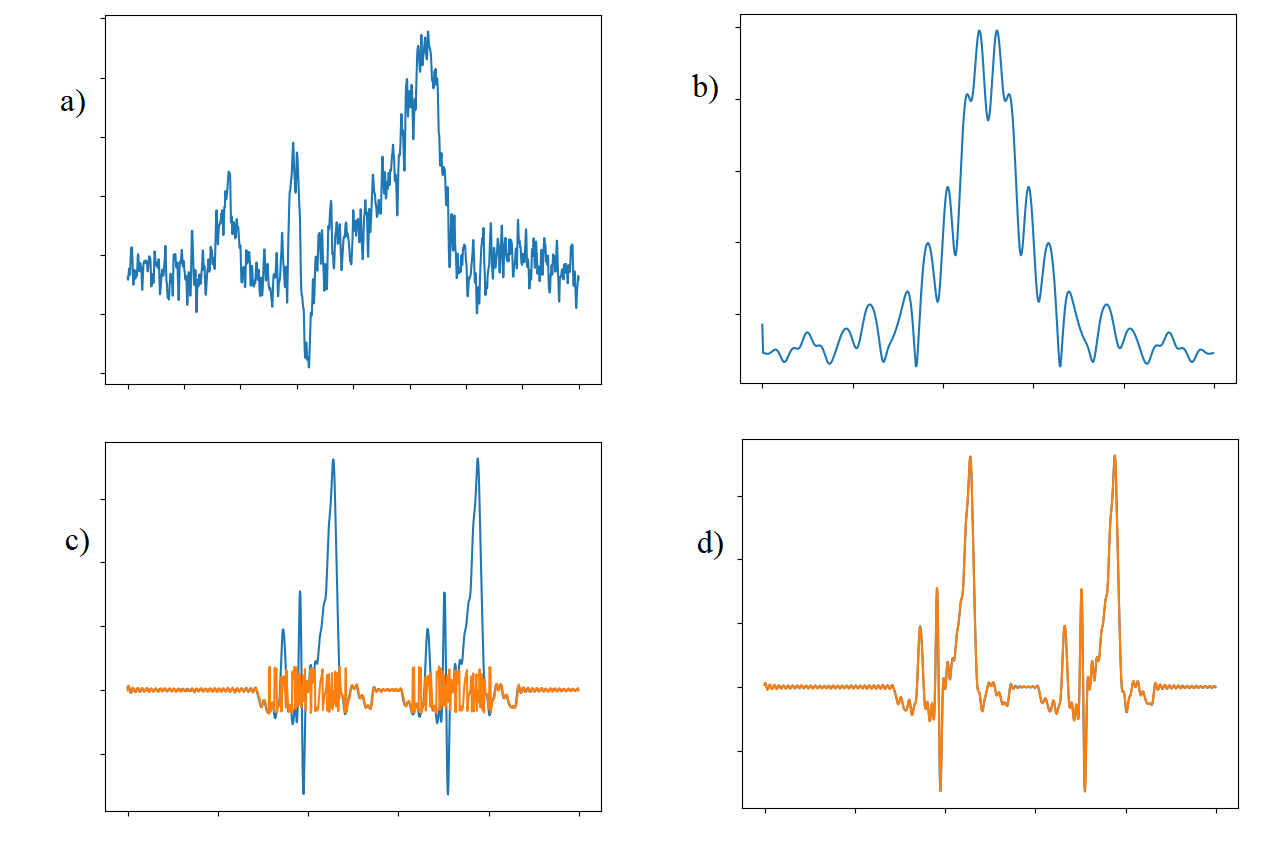}
	\end{minipage}
	\caption{The figure demonstrates the application of modulo sampling to ECG signals in a shift-invariant context. (a) The base pulse used, shown at the original sampling rate of $1000$ Hz. (b) The function $|H(\omega)|$, given these parameters, demonstrates that the correction filter is equivalent to $\frac{1}{H(\omega)}$. It is important to note that $|H(\omega)|$ is bounded away from zero, ensuring some stability in the correction process. (c) A segment of the input signal (in blue) alongside its modulo-folded version (in orange), illustrating the folding process. (d) Compares the original signal (in blue) with their reconstruction (in orange), underscoring the precise recovery of coefficients and thereby the accurate reconstruction of the ECG signal.}
	\label{fig:ecg}
\end{figure*}

\section{Conclusion}
\label{sec:conclusion}

In this study, we introduced an innovative approach for the recovery of SI signals through modulo sampling, highlighting the integration of a mixer and a LPF to produce a BL signal. This BL signal is then accurately reconstructed from its modulo samples if sampled above the Nyquist rate, building upon findings from prior research. The original SI signal is effectively retrieved from this BL counterpart. Our methodology enables recovery at any sampling rate below \(T\), which is the theoretical minimum rate for SI signals. Through extensive simulations, we evaluated the system's robustness against varying noise levels, showcasing the significant advantage of incorporating a mixer. This addition proves particularly beneficial for enhancing stability in cases involving unstable generators, thereby expanding the practical applications and effectiveness of modulo sampling in SI signal recovery.

\bibliographystyle{IEEEtran}
\bibliography{refs}

\end{document}